\documentclass[twoside]{article}
\usepackage{amsfonts, amsthm, amssymb}
\usepackage[mathcal]{eucal}

\newtheorem{theorem}{Theorem}[section]
\newtheorem{lemma}[theorem]{Lemma}
\newtheorem{prop}[theorem]{Proposition}
\newtheorem{cor}[theorem]{Corollary}

\def\l{\lambda}
\def\a{\alpha}
\def\b{\beta}
\def\g{\gamma}
\def\d{\delta}

\def\t{\tau}

\def\R{\mathbb{R}}

\def\C{\mathbb{C}}

\def\Z{\mathbb{Z}}
\def\P{\mathbb{P}}

\def\T{\mathbb{T}}

\def\M{\mathcal{M}}
\def\Mba{\M_{\beta, \alpha}}
\def\Hba{H_{\beta, \alpha}}

\def\P{\mathcal{P}}
\def\Pba{\P_{\beta, \alpha}}

\def\nba{\nabla_{b,a}}

\def\mbf2{\mathbf{2}}

\def\1N1{1 \leq k \leq N-1}

\begin{document}

\title{Nekhoroshev theorem for the periodic Toda lattice}
\author{Andreas Henrici\footnote{Supported in part by the Swiss National Science Foundation}
 \and Thomas Kappeler\footnote{Supported in part by the Swiss National Science Foundation, the programme SPECT and the European Community through the FP6 Marie Curie RTN ENIGMA (MRTN-CT-2004-5652)}}

\maketitle

\begin{abstract}
The periodic Toda lattice with $N$ sites is globally symplectomorphic to a two parameter family of $N-1$ coupled harmonic oscillators. The action variables fill out the whole positive quadrant of $\R^{N-1}$. We prove that in the interior of the positive quadrant as well as in a neighborhood of the origin, the Toda Hamiltonian is strictly convex and therefore Nekhoroshev's theorem applies on (almost) all parts of phase space.
\end{abstract}

\section{Introduction} \label{introduction}

Consider the periodic Toda lattice with period $N$ ($N \geq 2$), 
\begin{displaymath}
\dot{q}_n = \partial_{p_n} H_{Toda}, \quad \dot{p}_n = -\partial_{q_n} H_{Toda}, \quad n \in \Z
\end{displaymath}
where the (real) coordinates $(q_n, p_n)_{n \in \Z}$ satisfy $(q_{n+N}, p_{n+N}) = (q_n, p_n)$ for any $n \in \Z$ and the Hamiltonian $H_{Toda}$ is given by
\begin{displaymath}
  H_{Toda} = \frac{1}{2} \sum_{n=1}^N p_n^2 + \sum_{n=1}^N V(q_n - q_{n+1})
\end{displaymath}
with potential
\begin{equation} \label{potentialformel}
	V(x) = \g^2 e^{\d x} + V_1 x + V_2
\end{equation}
and constants $\g$, $\d$, $V_1$, $V_2 \in \R$ ($\g, \d \neq 0$). The Toda lattice has been introduced by Toda \cite{toda} and studied extensively in the sequel. It is an FPU lattice, i.e. a Hamiltonian system of particles in one space dimension with nearest neighbor interaction. Models of this type have been studied by Fermi-Pasta-Ulam [FPU]. In numerical experiments they found recurrent features for the lattices they considered. Despite an enormous effort from the physics and mathematics community in the past fifty years, by and large, these numerical experiments still defy an explanation. For a recent account of the fascinating history of the FPU problem, see e.g. \cite{beiz} or \cite{gall}. At least in the case of the periodic Toda lattice, the recurrent features can be fully accounted for. In fact, Flaschka \cite{fla1}, H\'enon \cite{henon}, and Manakov \cite{mana} independently proved that the periodic Toda lattice is integrable. In this paper, we show that on the open dense subset of the phase space where all action variables are strictly positive, the Nekhoroshev theorem \cite{nekh1, nekh2} applies. It means that the action variables of the Toda lattice vary slowly over an exponentially long time interval along solutions of a Hamiltonian system with Hamiltonian sufficiently close to $H_{Toda}$.

To continue, let us note that in (\ref{potentialformel}), without loss of generality, we can assume that $V_1 = V_2 = 0$. When expressed in the canonical coordinates $(\d q_j, \frac{1}{\d} p_j)_{1 \leq j \leq N}$, the Hamiltonian $H_{Toda}$ is, up to a scaling factor $\d^{-2}$, of the form
\begin{equation} \label{htoda}
H_{Toda} = \frac{1}{2} \sum_{n=1}^N p_n^2 + \alpha^2 \sum_{n=1}^N e^{q_n - q_{n+1}}
\end{equation}
where $a = |\g \d|$. Moreover, notice that the total momentum $\sum_{n=1}^N p_n$ is conserved. Hence the motion of the center of mass $\frac{1}{N} \sum_{n=1}^N q_n$ is linear and therefore unbounded. However, the orbits of the system relative to the center of mass all lie on tori. To describe these orbits, consider the relative coordinates $v_n := q_{n+1} - q_n$ ($1 \leq n \leq N-1$) and their canonically conjugate ones, $u_n := n\b - \sum_{j=1}^n p_k$ ($1 \leq n \leq N-1$), where $\b = \frac{1}{N} \sum_{j=1}^N p_n$. In the sequel, we view the Toda lattice as a two parameter family of integrable systems with the two parameters $\a>0$ and $\b \in \R$. For $\a>0$ and $\b \in \R$ arbitrary, denote by $\Hba$ the Toda Hamiltonian when expressed in the canonical coordinates $(v_k, u_k)_{\1N1} \in \R^{2N-2}$ and the parameters $\a$ and $\b$.

In \cite{ahtk2}, we proved the following result.

\begin{theorem} \label{bnftodatheoremintro}
The periodic Toda lattice admits Birkhoff coordinates. More precisely, there exist globally defined canonical coordinates $(x_k, y_k)_{\1N1} \in \R^{2N-2}$ so that for any $\b \in \R$ and $\a>0$, the Toda Hamiltonian $\Hba$, when expressed in these coordinates, takes the form $\frac{N\b^2}{2} + H_\a(I)$, where $H_\a(I)$ is a real analytic function of the action variables $I_k = (x_k^2 + y_k^2)/2$ ($\1N1$) alone.
\end{theorem}

In particular, Theorem \ref{bnftodatheoremintro} states that the action variables $(I_n)_{1 \leq n \leq N-1}$ are independent of $\b \in \R$ and $\a>0$. Note that each of the $N-1$ frequencies
	\[ \omega_i = \partial_{I_i} \left( \frac{N \b^2}{2} + H_\a(I) \right) = \partial_{I_i} H_\a(I)
\]
of the Toda lattice $\Hba$ is independent of the parameter $\b$.

The main result of this paper says that the Hamiltonian $H_\a$, introduced in Theorem \ref{bnftodatheoremintro}, is a convex function of the action variables $(I_k)_{\1N1}$:

\begin{theorem} \label{nekhtodatheorem}
In the open quadrant $\R_{>0}^{N-1}$, the Hamiltonian $H_\a$ is a strictly convex function of the action variables $(I_k)_{\1N1}$. More precisely, for any compact subset $U \subseteq \R_{>0}^{N-1}$ and any compact interval $[\a_1, \a_2] \subseteq \R_{>0}$, there exists $m>0$, such that 
\begin{equation} \label{hbamconv}
  \langle \partial_I^2 H_\a(I) \xi, \xi \rangle \geq m \| \xi \|^2, \qquad \forall \xi \in \R^{N-1}
\end{equation}
for any $I \in U$, any $\b \in \R$, and any $\a_1 \leq \a \leq \a_2$.
\end{theorem}

Theorem \ref{nekhtodatheorem} implies that Nekhoroshev's theorem holds for the Toda lattice on 
	\[ \mathcal{P}^\bullet := \{ (v,u) \in \R^{2N-2} | \, I_n(v,u) > 0 \; \forall \, 1 \leq n \leq N-1 \},
\]
an open and dense subset of $\R^{2N-2}$ by Theorem \ref{bnftodatheoremintro}.

\begin{cor}
For any $\b \in \R$ and $\a>0$, Nekhoroshev's theorem applies to (sufficiently small) Hamiltonian perturbations of the Toda Hamiltonian $\Hba$ on all of $\mathcal{P}^\bullet$. (See \cite{lochak}, \cite{lone}, \cite{nekh1}, \cite{nekh2}, \cite{poeschel0}, \cite{poeschel} for various versions of Nekhoroshev's theorem and their proofs.)
\end{cor}

In practice, it is difficult to verify for an integrable system with a given Hamiltonian $H$ whether the convexity (or steepness) condition of Nekhoroshev's theorem is satisfied as this condition refers to $H$, when expressed in action variables, and is \emph{not} invariant under canonical transformations. Typically one does not know the Hamiltonian as a function of the action variables explicitly enough to derive the convexity property.

To prove Theorem \ref{nekhtodatheorem}, we make use of the Birkhoff normal form of the Toda lattice $\Hba$ on $\R^{2N-2}$ near the elliptic fixed point $(v,u) = (0,0)$, established in \cite{ahtk3}.

\begin{theorem} \label{bnf4prop}
Let $\a>0$ be arbitrary. Near $I=0$, the function $H_\a(I)$, introduced in Theorem \ref{bnftodatheoremintro}, has an expansion of the form
\begin{equation} \label{halphataylor}
  N \a^2 + 2\a \sum_{k=1}^{N-1} s_k I_k + \frac{1}{4N} \sum_{k=1}^{N-1} I_k^2 + O(I^3),
\end{equation}
with $s_k = \sin \frac{k \pi}{N}$ for $\1N1$. In particular, the Hessian of $H_\a(I)$ at $I=0$ is given by
\begin{displaymath}
  \partial^2_I H_\a|_{I=0} = \frac{1}{2N} \textrm{Id}_{N-1}.
\end{displaymath}
\end{theorem}

As an immediate consequence of Theorem \ref{bnf4prop} we get

\begin{cor} \label{strictconvexlocal}
Near $I=0$, $H_\a(I)$ is strictly convex for any $\a>0$.
\end{cor}

Outside of $I=0$, we argue differently. For any $\a>0$, consider the frequency map
	\[ \R_{\geq 0}^{N-1} \to \R^{N-1}, I \mapsto \omega(I;\a) := \partial_I H_\a.
\]
In view of Corollary \ref{strictconvexlocal}, Theorem \ref{nekhtodatheorem} follows once we can show that the frequency map is \emph{nondegenerate} on all of $\R_{>0}^{N-1}$. Note that the property of being nondegenerate is \emph{invariant} under coordinate transformations, an observation used in a crucial way in the sequel. In section \ref{freqtoda}, we prove that on $\mathcal{P}^\bullet$, the frequencies can be expressed in terms of periods of a certain Abelian differential of the $2^{nd}$ kind. To show that the frequency map is nondegenerate on $\mathcal{P}^\bullet$ we use, in addition to Theorem \ref{bnftodatheoremintro}, a version of Krichever's theorem (Theorem \ref{krichevertheorem}) suited for applications to the Toda frequencies. In \cite{krich}, Krichever stated his result concerning the period map of certain Abelian differentials for Hill's curve. Bikbaev and Kuksin presented a proof of this result in \cite{biku}. In section \ref{krichever} we apply their scheme of proof to prove the version of Krichever's theorem needed for our purposes. In section \ref{mainsection}, we prove Theorem \ref{nekhtodatheorem}.



\emph{Acknowledgements:} It is a great pleasure to thank Sergei Kuksin for valuable discussions.

\section{Preliminaries} \label{prelim}

To prove the integrability of the Toda lattice, Flaschka introduced the (noncanonical) coordinates (cf. \cite{fla1})
\begin{displaymath}
b_n := -p_n \in \R, \quad a_n := \alpha e^{\frac{1}{2} (q_n - q_{n+1})} \in \R_{>0} \quad (n \in \Z).
\end{displaymath}
These coordinates describe the motion of the Toda lattice relative to the center of mass. They are related to the relative coordinates defined in section \ref{introduction} as follows
	\[ ((u_n, v_n)_{1 \leq n \leq N-1}, \b, \a) \mapsto (b_n, a_n)_{1 \leq n \leq N}
\]
with $a_n = \a \exp \left( -\frac{1}{2} v_n \right)$ ($1 \leq n \leq N-1$), $a_N = \a \exp \left( \frac{1}{2} \sum_{k=1}^{N-1} v_k \right)$, $b_1 = u_1 - \b$, $b_n = u_n - u_{n-1} - \b$ ($2 \leq n \leq N-1$), and $b_N = -u_{N-1} - \b$. In the sequel we will work with the coordinates $(b_n, a_n)_{1 \leq n \leq N}$ rather than the relative coordinates $(u_n, v_n)_{1 \leq n \leq N-1}$. 
In these coordinates, the Hamiltonian $H_{Toda}$ takes the simple form
\begin{equation} \label{htodaflaschka}
	H_{Toda} = \frac{1}{2} \sum_{n=1}^N b_n^2 + \sum_{n=1}^N a_n^2,
\end{equation}
and the equations of motion are
\begin{equation} \label{flaeqn}
\left\{ \begin{array}{lllll}
 \dot{b}_n & = & a_n^2 - a_{n-1}^2 \\
 \dot{a}_n & = & \frac{1}{2} a_n (b_{n+1} - b_n)
\end{array} \right. \quad (n \in \Z).
\end{equation}
Note that $(b_{n+N}, a_{n+N}) = (b_n, a_n)$ for any $n \in \Z$, and $\prod_{n=1}^N a_n = \alpha^N$. Hence we can identify the sequences $(b_n)_{n \in \Z}$ and $(a_n)_{n \in \Z}$ with the vectors $(b_n)_{1 \leq n \leq N} \in \R^N$ and $(a_n)_{1 \leq n \leq N} \in \R_{>0}^N$. The phase space of the system (\ref{flaeqn}) is then given by
$$ \M := \R^N \times \R_{>0}^N, $$
and it turns out that (\ref{flaeqn}) is a Hamiltonian system with a nonstandard Poisson structure $J$ found by Flaschka \cite{fla1} (cf. \cite{ahtk1}). This Poisson structure is degenerate and admits the two Casimir functions\footnote{A smooth function $C: \M \to \R$ is a Casimir function for $J$ if $dC(x) \in$ Ker$J(x)$ for any $x \in \M$.}
\begin{equation} \label{casimirdef}
C_1 := -\frac{1}{N} \sum_{n=1}^N b_n \quad \textrm{and} \quad C _2 := \left( \prod_{n=1}^N a_n \right)^\frac{1}{N}.
\end{equation}
Let $\Mba := \{ (b,a) \in \M : (C_1, C_2) = (\b, \a) \}$ denote the level set of $(C_1, C_2)$ at $(\b, \a) \in \R \times \R_{>0}$. As $C_1$ and $C_2$ are real analytic on $\M$ and the gradients $\nba C_1$ and $\nba C_2$ are linearly independent everywhere on $\M$, the sets $\Mba$ are real analytic submanifolds of $\M$ of (real) codimension two. Furthermore, the pullback of the Poisson structure $J$ to $\Mba$, is nondegenerate everywhere on $\Mba$ and therefore induces a symplectic structure on $\Mba$. In this way, we obtain a symplectic foliation of $\M$ with $\Mba$ being the symplectic leaves. By a slight abuse of notation with respect to the definition made in section \ref{introduction}, we denote by $\Hba$ the restriction of the Hamiltonian $H_{Toda}$ to $\Mba$.

As a model space for the construction of canonical Cartesian coordinates on $\M$, we introduced in \cite{ahtk2} the space $\P := \R^{2(N-1)} \times \R \times \R_{>0}$, foliated by the leaves $\Pba := \R^{2(N-1)} \times \{ \beta \} \times \{ \a \}$ which are endowed with the standard symplectic structure. Denote by $J_0$ the degenerate Poisson structure on $\P$ having $\Pba$ as its symplectic leaves with standard symplectic structure and the coordinates $\b$ and $\a$ as its Casimirs. In \cite{ahtk2} we proved the following theorem which describes in more detail the results stated in Theorem \ref{bnftodatheoremintro}.

\begin{theorem} \label{sumthm}
There exists a map
\begin{displaymath}
\begin{array}{ccll}
 \Phi: & (\M, J) & \to & (\P, J_0) \\
 & (b,a) & \mapsto & ((x_n, y_n)_{1 \leq n \leq N-1}, C_1, C_2)
\end{array}
\end{displaymath}
with the following properties:
\begin{itemize}
  \item[(i)] $\Phi$ is a real analytic diffeomorphism.
  \item[(ii)] $\Phi$ is canonical, i.e. it preserves the  Poisson brackets. In particular, the symplectic foliation of $\M$ by $\Mba$ is trivial.
  \item[(iii)] The coordinates $(x_n, y_n)_{1 \leq n \leq N-1}, C_1, C_2$ are
  global Birkhoff coordinates for the periodic Toda lattice, i.e. the transformed Toda Hamiltonian $\hat{H} = H \circ \Phi^{-1}$
  is a function of the actions $I_n := (x_n^2 + y_n^2)/2$ $(1 \leq n \leq N-1)$ and $C_1, C_2$ alone. It is of the form $\frac{N \b^2}{2} + H_\a(I)$.
\end{itemize}
\end{theorem}

As an immediate consequence of Theorem \ref{sumthm} one gets
\begin{cor} \label{mbacor}
For any $\b \in \R$, $\a>0$, the set
	\[ \Mba^\bullet = \{ (b,a) \in \Mba | I_n(b,a)  > 0 \; \forall \, 1 \leq n \leq N-1 \}
\]
is open and dense in $\Mba$.
\end{cor}


For later use we describe an important ingredient in the proof of Theorem \ref{sumthm}. For any $(b,a) \in \M$ denote by $L^+(b,a)$ and $L^-(b,a)$ the symmetric $N \times N$-matrices defined by
\begin{equation} \label{jacobi}
L^\pm(b,a) := \left( \begin{array}{ccccc}
b_1 & a_1 & 0 & \ldots & \pm a_N \\
a_1 & b_2 & a_2 & \ddots & \vdots \\
0 & a_2 & b_3 & \ddots & 0 \\
\vdots & \ddots & \ddots & \ddots & a_{N-1} \\
\pm a_N & \ldots & 0 & a_{N-1} & b_N \\
\end{array} \right)
\end{equation}
and by $B$ the skew-symmetric $N \times N$-matrix
\begin{displaymath}
B = \left( \begin{array}{ccccc}
0 & a_1 & 0 & \ldots & -a_N \\
-a_1 & 0 & a_2 & \ddots & \vdots \\
0 & -a_2 & \ddots & \ddots & 0 \\
\vdots & \ddots & \ddots & \ddots & a_{N-1} \\
a_N & \ldots & 0 & -a_{N-1} & 0 \\
\end{array} \right).
\end{displaymath}
Flaschka \cite{fla1} observed that system (\ref{flaeqn}) admits the Lax pair formulation
	\[ \dot{L}^+= \partial_t L^+ = [B, L^+].
\]
As the flow of $\dot{L}^+ = [B, L^+]$ is isospectral, the eigenvalues of $L^+$ are conserved quantities of the Toda lattice. 
We need some results about the spectral theory of periodic Jacobi matrices (\ref{jacobi}).
For $(b,a) \in \M$, consider for any complex number $\l$ the difference equation
\begin{equation} \label{diff}
a_{k-1} y(k-1) + b_k y(k) + a_k y(k+1) = \l y(k) \quad (k \in \Z)
\end{equation}
associated to $L(b,a)$.
The two fundamental solutions $y_1(\cdot, \l)$ and $y_2(\cdot, \l)$ of (\ref{diff}) are defined by the standard initial conditions $y_1(0, \l) = 1$, $y_1(1, \l) = 0$ and $y_2(0, \l) = 0$, $y_2(1, \l) = 1$. By solving (\ref{diff}) recursively, one sees that for any $k$, $y_i(k,\l)$ $(i=1,2)$ is a polynomial in $\l$. Denote by $\Delta(\l) \equiv \Delta(\l, b, a)$ the \emph{discriminant} of (\ref{diff}), defined by
\begin{equation} \label{discrdef}
\Delta(\l) := y_1(N, \l) + y_2(N+1, \l).
\end{equation}
By Floquet theory, for any $\l \in \R$, equation (\ref{diff}) admits a periodic respectively antiperiodic solution of period $N$ iff the discriminant $\Delta_\l \equiv \Delta(\l)$ satisfies $\Delta_\l=2$ respecitvely $\Delta_\l=-2$. 
It follows that $\Delta_\l \mp 2$ admits a product representation of the form
\begin{displaymath}
  \Delta_\l \mp 2 = \alpha^{-N} \prod_{j=1}^{N} (\l - \l_j^\pm),
\end{displaymath}
where $(\l_j^\pm)_{1 \leq j \leq N}$ are the eigenvalues of $L^\pm(b,a)$. They are real valued and we list them in increasing order and with their algebraic multiplicities. Hence
\begin{equation} \label{delta2lrepr}
  \Delta_\l^2 - 4 = \alpha^{-2N} \prod_{j=1}^{2N} (\l - \l_j),
\end{equation}
where $(\l_j)_{1 \leq j \leq 2N}$ is the combined sequence of the eigenvalues $(\l_j^+)_{1 \leq j \leq N}$ and $(\l_j^-)_{1 \leq j \leq N}$ listed in increasing order. One can show that
\begin{equation} \label{spectgeordnet}
	\l_N^+ > \l_N^- \geq \l_{N-1}^- > \l_{N-1}^+ \geq \l_{N-2}^+ > \l_{N-2}^- \geq \l_{N-3}^- > \l_{N-3}^+ \geq \ldots
\end{equation}
Again by Floquet theory, one sees that $(\l_j)_{1 \leq j \leq 2N}$ are the eigenvalues of the $2N \times 2N$ Jacobi matrix $L^+((b,b),(a,a))$. Since $\Delta_\l$ is a polynomial of degree $N$, $\dot{\Delta}_\l = \partial_\l \Delta_\l$ is a polynomial of degree $N-1$. It admits a product representation of the form
\begin{equation} \label{dotlrepr}
  \dot{\Delta}_\l = N \a^{-N} \prod_{k=1}^{N-1} (\l - \dot{\l}_k),
\end{equation}
where the zeroes $(\dot{\l}_n)_{1 \leq n \leq N-1}$ of $\dot{\Delta}_\l$ are all real valued and are listed in increasing order. They satisfy $\l_{2n} \leq \dot{\l}_n \leq \l_{2n+1}$ for any $1 \leq n \leq N-1$. The open intervals $(\l_{2n}, \l_{2n+1})$ are referred to as the \emph{$n$-th spectral gap} and $\gamma_n := \l_{2n+1} - \l_{2n}$ as the \emph{$n$-th gap length}. Note that $|\Delta_\l| > 2$ on the spectral gaps. We say that the $n$-th gap is \emph{open} if $\gamma_n > 0$ and \emph{collapsed} otherwise. The set of elements $(b,a) \in \M$ for which the $n$-th gap is collapsed is denoted by $D_n$,
\begin{equation} \label{dndef}
D_n := \{ (b,a) \in \M: \gamma_n = 0 \}.
\end{equation}
Using that $\gamma_n^2$ (unlike $\gamma_n$) is a real analytic function on $\M$, it can be shown that $D_n$ is a real analytic submanifold of $\M$ of codimension $2$ (cf. \cite{kapo} for a similar statement in the case of Hill's operator). Moreover, one can show that for any $(b,a) \in \M$ and any $1 \leq n \leq N-1$, $\g_n(b,a) = 0$ iff $I_n(b,a) = 0$ - see \cite{ahtk1} for details. Hence for any $\b \in \R$ and $\a>0$, the set $\Mba^\bullet$, introduced in Corollary \ref{mbacor}, satisfies $\Mba^\bullet = \Mba \setminus \cup_{n=1}^{N-1} D_n$.

Finally, we remark that the zeros $(\l_j)_{1 \leq j \leq 2N}$ and $(\dot{\l}_k)_{1 \leq k \leq N-1}$ of $\Delta^2_\l-4$, respectively $\dot{\Delta}_\l$, satisfy the following relation
\begin{equation} \label{lambdadotlambda}
  \sum_{k=1}^{N-1} \dot{\l}_k = \frac{N-1}{2N} \sum_{j=1}^{2N} \l_j.
\end{equation}
To prove (\ref{lambdadotlambda}), one computes the $\l$-derivative of
\begin{displaymath}
  \Delta_\l = \pm 2 + \a^{-N} \prod_{j=1}^N (\l - \l_j^\pm)
\end{displaymath}
and compares the coefficients of the expansions of $\dot{\Delta}_\l$, obtained in this way, with (\ref{dotlrepr}).

\section{Krichever's theorem} \label{krichever}

In this section, we present a version of Krichever's theorem suited to prove Theorem \ref{nekhtodatheorem}. Krichever's theorem concerns the period map of certain meromorphic differentials of a hyperelliptic Riemann surface. In \cite{krich}, Krichever stated his theorem for a parameter family of hyperelliptic curves having the property that one of the ramification points is at infinity. In the version we need we have to consider a parameter family of hyperelliptic curves with no ramification points at infinity. 

Let $E = (E_1, \ldots, E_{2N})$ be a sequence of distinct, but otherwise arbitrary real numbers which we list in increasing order, $E_1 < E_2 < \ldots < E_{2N-1} < E_{2N}$. Introduce
	\[ R(\l) = \prod_{i=1}^{2N} (\l - E_i), \; \l \in \C
\]
and denote by $\mathcal{C}_E$ the affine curve
	\[ \mathcal{C}_E = \{ (\l,w) \in \C^2 : w^2 = \sigma R(\l) \},
\]
where $\sigma \in \R_{>0}$ is a scaling parameter. In our application to the Toda lattice it will be given by $\a^{-2N}$. Then $\mathcal{C}_E$ is a two-sheeted curve with ramification points $(E_i, 0)_{1 \leq i \leq 2N}$, identified with $E_i$ in the sequel. By $\Sigma_E$ we denote the Riemann surface obtained from $\mathcal{C}_E$
 by adding the two (unramified) points at infinity, $\infty^+$ and $\infty^-$, one on each of the two sheets. The sheet of $\Sigma_E$ which contains $\infty^-$ is also referred to as the canonical sheet and denoted by $\Sigma_E^c$. It is characterized by
	\[ w = \sqrt{R(\l - i0)} < 0 \quad \forall \; \l \in \R \textrm{ with } \l > E_{2N}.
\]

The variable $z$ around $z=0$ gives a complex chart in a neighborhood of $\infty^+$ or $\infty^-$ of $\Sigma_E$ via the substitution $\l = \frac{1}{z}$. By construction, these charts at $\infty^+$ and $\infty^-$ are defined in a unique way and are referred to as standard charts of $\infty^\pm$. 

It is convenient to introduce the projection $\pi \equiv \pi_E: \mathcal{C}_E \to \C$ onto the $\l$-plane, i.e. $\pi_E(\l,w) = \l$ and its extension to a map $\pi_E: \Sigma_E \to \C \cup \{ \infty \}$, where $\pi_E(\infty^\pm) = \infty$. Denote by $(c_k)_{\1N1}$ the cycles on the canonical sheet of $\mathcal{C}_E$ so that $\pi(c_k)$ is a counterclockwise oriented closed curve in $\C$, containing in its interior the two ramification points $E_{2k}$ and $E_{2k+1}$, whereas all other ramification points are outside of $\pi(c_k)$. The following result is straightforward to prove.

\begin{lemma} \label{abeldiffexist}
There exist Abelian differentials $\Omega_1$ and $\Omega_2$ on $\Sigma_E$, uniquely determined by the following properties:
\begin{itemize}
	\item[(i)] $\Omega_1$ and $\Omega_2$ are holomorphic on $\Sigma_E$ except at $\infty^+$ and $\infty^-$ where in the standard charts, $\Omega_i$ admit an expansion of the form
	\[ \Omega_1 = \mp \left( -\frac{1}{z} + e_1 + O(z) \right) dz \left( = \mp \left( \frac{1}{\l} - \frac{e_1}{\l^2} + O \left( \frac{1}{\l^3} \right) \right) d\l \right)
\]
and
	\[ \Omega_2 = \mp \left( -\frac{1}{z^2} + f_1 + O(z) \right) dz \left( = \mp \left( 1 + O \left( \frac{1}{\l^2} \right) \right) d\l \right).
\]
\item[(ii)] $\Omega_1$ and $\Omega_2$ satisfy the normalization condtions
\begin{equation} \label{omegainorm}
	\int_{c_k} \Omega_i = 0 \quad \forall \, \1N1, \; i=1,2.
\end{equation}
\end{itemize}
On $\mathcal{C}_E \setminus E$, $\Omega_1$ and $\Omega_2$ take the form $\Omega_i = \frac{\chi_i(\l)}{\sqrt{R(\l)}} \, d\l$ ($i=1,2$), where $\chi_i(\l)$ are polynomials in $\l$ of the form $\chi_1(\l) = \l^{N-1} + e \l^{N-2} + \ldots$ and 
$\chi_2(\l) = \l^N + f \l^{N-1} + \ldots$, with $f = -\frac{1}{2} \sum_{n=1}^{2N} E_n$. In particular, $\Omega_1$ and $\Omega_2$ do not depend on the scaling parameter $\sigma$.
\end{lemma}

Denote by $(d_k)_{\1N1}$ pairwise disjoint cycles on $\mathcal{C}_E \setminus E$ so that for any $1 \leq n,k \leq N-1$, the intersection indices with the cycles $(c_n)_{1 \leq n \leq N-1}$ with respect to the orientation on $\Sigma_E$, induced by the complex structure, are given by $c_n \circ d_k = \delta_{nk}$. In order to be more precise, choose
the cycles $d_k$ in such a way that (i) the projection $\pi_E(d_k)$ of $d_k$ is a smooth, convex counterclockwise oriented curve in $\C \setminus ((E_1, E_{2k}) \cup (E_{2k+1}, \infty))$ and (ii) the points of $d_k$ whose projection by $\pi_E$ onto the $\l$-plane have a negative imaginary part lie on the canonical sheet of $\Sigma_E$. For any $\1N1$, introduce the $d_k$-periods of $\Omega_1$ and $\Omega_2$,
\begin{equation} \label{ukvkdef}
	U_k := \int_{d_k} \Omega_1; \qquad V_k := \int_{d_k} \Omega_2
\end{equation}
and for any $p \in \mathcal{C}_E$ define the Abel integrals ($i=1,2$)
	\[ J_i(p) = \frac{1}{2} \int_{\gamma_p} \Omega_i,
\]
where $\gamma_p$ is any path in $\mathcal{C}_E$ from $p_*$ to $p$. The map $\iota: \mathcal{C}_E \to \mathcal{C}_E$, $p \mapsto p_*$ interchanges the two sheets of $\mathcal{C}_E$,
	\[ p_* = (\l,-w) \quad \forall \, p=(\l,w) \in \mathcal{C}_E.
\]
Note that for any $i=1,2$, the function $p \mapsto J_i(p)$ is multi-valued. Actually, $J_i(p)$ is well defined up to half periods of $\Omega_i$. Hence locally it is a well defined smooth function. In particular, its differential $d J_i$ is well defined. Note that for $i=1,2$ and $1 \leq n \leq 2N$, zero is one of the possible values of $J_i(E_n)$. For any $p \in \mathcal{C}_E$, denote by $\gamma_p^0$ a path on $\mathcal{C}_E$ from $E_{2N} \equiv (E_{2N},0)$ to $p$ and define $\gamma_p$ to be the path from $p_*$ to $p$ obtained by concatenating $-\iota \left( \gamma_p^0 \right)$ and $\gamma_p^0$. Here $-\iota \left( \gamma_p^0 \right)$ denotes the path from $p_*$ to $E_{2N}$ obtained by reversing the orientation of $\iota \left( \gamma_p^0 \right)$ and $\iota(\gamma_p^0)$ is the path obtained by applying to $\gamma_p^0$ the map $\iota$. In Lemma \ref{omegaiji} we state the properties of $\Omega_i$ and $J_i$ needed in the sequel. 

\begin{lemma} \label{omegaiji}
\begin{itemize}
	\item[(i)] The differential forms $\Omega_1$ and $\Omega_2$ are odd with respect to the map $\iota$, i.e. the pullback $\iota^* \Omega_i$ of $\Omega_i$ satisfies $\iota^* \Omega_i = -\Omega_i$.
	\item[(ii)] For $i=1,2$, 
	\[ \frac{1}{2} \int_{-\iota \left( \gamma_P^0 \right) \circ \gamma_P^0} \Omega_i = \int_{\gamma_P^0} \Omega_i.
\]
	\item[(iii)] When expressed in the local coordinate $\l$, on each of the two sheets, $\int_{E_{2N}}^\l \Omega_i$ admits an asymptotic expansion as $\l \to \infty$ ($\l$ real) of the form
\begin{equation} \label{intomega1}
	\int_{E_{2N}}^\l \Omega_1 = \mp \left( \log \l + e_0 + e_1 \frac{1}{\l} + \ldots \right)
\end{equation}
and
\begin{equation} \label{intomega2}
	\int_{E_{2N}}^\l \Omega_2 = \mp \left( \l + f_0 + \ldots \right),
\end{equation}
where $e_0$ and $e_1$ are real valued.
\end{itemize}
\end{lemma}

\begin{proof}
(i) Let $1 \leq i \leq 2$. The claimed identity $\iota^* \Omega_i = -\Omega_i$ follows from the uniqueness of the differential $\Omega_i$ stated in Lemma \ref{abeldiffexist}, as $-\iota^* \Omega_i$ is a meromorphic differential which is holomorphic on $\mathcal{C}_E$ and satisfies the same asymptotics at $\infty^\pm$ and the same normalization condition (\ref{omegainorm}) as the differential $\Omega_i$. (ii) In view of statement (i) we conclude that for any $p \in \mathcal{C}_E$,
	\[ \frac{1}{2} \int_{-\iota \left( \gamma_P^0 \right) \circ \gamma_P^0} \Omega_i = \frac{1}{2} \left( -\int_{\iota \left( \gamma_P^0 \right)} \Omega_i + \int_{\gamma_P^0} \Omega_i \right) = \int_{\gamma_P^0} \Omega_i.
\]

(iii) The stated asymptotics follow from the asymptotics of $\Omega_i$ of Lemma \ref{abeldiffexist}. The claim of $e_0$ and $e_1$ being real follows from the assumption that $E_1, \ldots, E_{2N}$ are real and that for $\l$ real with $\l > E_{2N}$, one has $R(\l) > 0$.
\end{proof}

To state the main result of this section, introduce the extended period map, defined on the space of sequences $E = (E_1 < \ldots < E_{2N})$ as follows
\begin{equation} \label{mathcalfdef}
	\mathcal{F}: E \mapsto ((U_i, V_i)_{1 \leq i \leq N-1}, e_1, e_0),
\end{equation}
where $e_1$ and $e_0$ are the coefficients in the asymptotic expansion (\ref{intomega1}). It is straightforward to see that $\mathcal{F}$ is a smooth map with values in $\R^{2N}$. The version of Krichever's theorem needed for our purposes is the following one.

\begin{theorem} \label{krichevertheorem}
At each point $E = (E_1 < \ldots < E_{2N})$, the map $\mathcal{F}$ is a local diffeomorphism, i.e. the differential $d_E \mathcal{F}: \R^{2N} \to \R^{2N}$ of $\mathcal{F}$ at $E$ is a linear isomorphism.
\end{theorem}

The proof of Theorem \ref{krichevertheorem} follows the scheme used in \cite{biku} to prove Krichever's theorem. First we need to derive some auxiliary results. For any $1 \leq i \leq 2$, denote by $N_{\Omega_i}$ the set of zeroes of $\Omega_i$ and by $N_{\chi_i}$ the set of zeroes of the polynomials $\chi_i$, where in both cases the zeroes are listed with their multiplicities. Note that $| N_{\chi_1} | = N-1$ and $| N_{\chi_2} | = N$, whereas for $i=1,2$
	\[ | N_{\Omega_i} | \leq 2 | N_{\chi_i} |.
\]
The following result is due to \cite{biku}.

\begin{lemma} \label{chi1chi2}
The zero sets $N_{\chi_i}$ and $N_{\Omega_i}$ have the following properties:
\begin{itemize}
	\item[(i)] All zeroes of $N_{\chi_1}$ are simple and real, and $N_{\chi_1} \cap \{ E_1, \ldots, E_{2N} \} = \emptyset$. Moreover, $N_{\Omega_1} = \pi_E^{-1}(N_{\chi_1})$ and $|N_{\Omega_1}| = 2N-2$.
	\item[(ii)] All zeroes of $N_{\chi_2}$ are simple except possibly one which then has multiplicity two. Furthermore,
\begin{displaymath} 
	|N_{\chi_2} \setminus \{ E_1, \ldots, E_{2N} \}| \geq N-1 \quad \textrm{and} \quad |N_{\Omega_2} \setminus \{ E_1, \ldots, E_{2N} \}| \geq 2N-2.
\end{displaymath}
	\item[(iii)] $N_{\chi_1} \cap N_{\chi_2} = \emptyset$, and hence $N_{\Omega_1} \cap N_{\Omega_2} = \emptyset$ as well.
\end{itemize}
\end{lemma}

\begin{proof}[Proof of Lemma \ref{chi1chi2}]
The statements about the zero sets $N_{\Omega_i}$ of $\Omega_i$ are easily obtained from the ones about $N_{\chi_i}$ in view of the representation $\Omega_i = \chi_i(\l) / \sqrt{R(\l)} \,\, d\l$ and the property that $\Omega_i$ has a pole at $\infty^+$ and $\infty^-$. Hence we prove only the claimed statements for $N_{\chi_i}$.

By the normalization condition (\ref{omegainorm}), for any $\1N1$, $\chi_1(\l)$ has at least one real zero $\tau_{1,k}$ satisfying $E_{2k} < \tau_{1,k} < E_{2k+1}$. As $\chi_1(\l)$ is a polynomial of degree $N-1$, it follows that all zeroes $\tau_{1,k}$ are simple and that
	\[ N_{\chi_1} = \{ \tau_{1,k} | \1N1 \}.
\]
In particular $N_{\chi_1} \cap \{ E_1, \ldots, E_{2N} \} = \emptyset$. Similarly, (\ref{omegainorm}) implies that for any $\1N1$, $\chi_2(\l)$ has at least one real zero $\tau_{2,k}$ satisfying $E_{2k} < \tau_{2,k} < E_{2k+1}$. As $\chi_2(\l)$ is a polynomial of degree $N$,
	\[ N_{\chi_2} \setminus \{ \tau_{2,k} | \1N1 \}
\]
consists of one point $\tau_0 \in \C$. It is not excluded that $\tau_0$ coincides with one of the zeroes $(\tau_{2,k})_{\1N1}$. 
In any case, $|N_{\chi_2} \cap \{ E_1, \ldots, E_{2N} \}| \leq 1$. 
It remains to prove (iii). Assume that $\tau$ is a common zero of $\chi_1(\l)$ and $\chi_2(\l)$, i.e. $\tau \in N_{\chi_1} \cap N_{\chi_2}$. Then there exists $\1N1$ with $E_{2k} < \tau < E_{2k+1}$. As all roots of $\chi_1(\l)$ are \emph{simple}, one has $\chi_1'(\tau) \neq 0$ ($' = \frac{d}{d \l}$). Hence we can choose the real parameter $\xi$ in such a way that the polynomial $\chi_2 + \xi \chi_1$ has a double root at $\tau$. Indeed, for $\xi = -\chi_2'(\tau) / \chi_1'(\tau)$ one has $\chi_2(\tau) + \xi \chi_1(\tau) = 0$ and $\chi_2'(\tau) + \xi \chi_1'(\tau) = 0$. 
As
	\[ \int_{c_j} (\chi_2(\l) + \xi \chi_1(\l)) / \sqrt{R(\l)} \, d\l = 0 \quad \forall \; 1 \leq j \leq N-1,
\]
the $N$ roots of $\chi_2 + \xi \chi_1$ are given by $\t$ and $(\t_{\xi,j})_{j \neq k}$, where $\t$ is a double root and for any $j \neq k$, $E_{2j} < \t_{\xi,j} < E_{2j+1}$ is simple. Therefore, $\chi_2(\l) + \xi \chi_1(\l)$ does not change sign in the interval $[E_{2k}, E_{2k+1}]$, contradicting the normalization condition $\int_{c_k} (\chi_2(\l) + \xi \chi_1(\l)) \, d\l = 0$. Hence $\chi_1$ and $\chi_2$ have no zero in common, as claimed.
\end{proof}


\begin{proof}[Proof of Theorem \ref{krichevertheorem}]
Assume that Theorem \ref{krichevertheorem} does not hold. Then there exists a smooth $1$-parameter family $E(\t)$, $-1 < \t < 1$, so that for some $1 \leq n \leq 2N$, $\delta E_n \equiv \partial_\t |_{\t=0} E_n(\t) \neq 0$, but
\begin{eqnarray*}
  U(\tau) & = & U(0) + O(\tau^2), \quad V(\tau) = V(0) + O(\tau^2),\\
  e_0(\t) & = & e_0(0) + O(\t^2), \quad e_1(\t) = e_1(0) + O(\t^2).
\end{eqnarray*}
We will now prove that $\delta E_k = 0$ for any $1 \leq k \leq 2N$, leading to a contradiction. As above, we introduce for $p \in \mathcal{C}_{E(\t)}$ the multi-valued functions $J_i(p,\t)$, defined up to half periods of $\Omega_i(\t)$,
\begin{displaymath} \label{abelintdef}
  J_i(p, \t) = \frac{1}{2} \int_{\gamma_p} \Omega_i(\t)
\end{displaymath}
where $\Omega_i(\t)$ denote the Abelian differentials of Lemma \ref{abeldiffexist}, corresponding to the Riemann surface $\Sigma_{E(\t)} = \mathcal{C}_{E(\t)} \cup \{ \infty^+, \infty^- \}$. By Lemma \ref{omegaiji} (ii), $J_i(p,\t) = \int_{\g_p^0} \Omega_i(\t)$. In particular, the differential $dJ_i(p,\t)$ is well defined and equals the restriction of $\Omega_i(\t)$ to $\mathcal{C}_{E(\t)}$. Near any point $p=(\l,w) \in \mathcal{C}_E \setminus E$, $\l$ is a local coordinate. This remains true for $\t$ sufficiently close to $0$, and hence for any $p \in \mathcal{C}_E \setminus E$ we can define $(i=1,2)$
\begin{equation} \label{deltajdef}
	\delta J_i(p) := \partial_\t|_{\t=0} J_i(p,\t).
\end{equation}
By Lemma \ref{propdtomega1} below, $\delta J_1$ is single-valued, extends to a meromorphic function on $\Sigma_E$ and is holomorphic on $\Sigma_E \setminus E$. At a ramification point $E_k$, the function $\delta J_1$ might have a simple pole with residue of the form $r_1(k) \d E_k$, where $r_1(k) \neq 0$. But by Proposition \ref{dj1equiv0} below, $\delta J_1 \equiv 0$ and hence, in particular, $\d E_k = 0$ for any $1 \leq k \leq 2N$. This contradicts the assumption made above that $\d E_n \neq 0$.
\end{proof}

It remains to prove Lemma \ref{propdtomega1} and Proposition \ref{dj1equiv0} mentioned in the proof of Theorem \ref{krichevertheorem}. Throughout the rest of this section we assume that the $1$-parameter family $E(\t)$ satisfies the assumption made in the proof of Theorem \ref{krichevertheorem}.

\begin{lemma} \label{propdtomega1}
The functions $\d J_1$ and $\d J_2$ defined by (\ref{deltajdef}) are single-valued and extend to meromorphic functions on $\Sigma_E$. They are holomorphic on $\Sigma_E \setminus E$. At the ramification points $(E_n)_{1 \leq n \leq 2N}$, they might have poles of order $1$ with residue of the form ($i=1,2; 1 \leq n \leq 2N$)
	\[ \textrm{Res}_{p=E_n} \d J_i = r_i(n) \d E_n
\]
where for $i=1$, $r_1(n) \neq 0$ for any $1 \leq n \leq 2N$. Moreover $\d J_1$ has a zero of order $2$ at $\infty^\pm$.
\end{lemma}

\begin{proof}
Let $1 \leq i \leq 2$ be given. Although the integral $J_i(p,\t)$, defined for $p \in \mathcal{C}_E$, is multi-valued in the sense that it is defined only up to half-periods of $\Omega_i(\t)$, the derivative $\partial_\t|_{\t=0} J_i(p,\t)$ is single-valued, since by assumption, the periods of $\Omega_i(\t)$ are constant up to $O(\t^2)$. To simplify notation we write $\Omega_i$ instead of $\Omega_i(\t)$ and $E_n$ instead of $E_n(\t)$. To see that $\d J_i(\t)$ extends meromorphically to any branching point $E_n$, note that near $E_n$, $\Omega_i$ admits an expansion in terms of $z = (\l - E_n)^{1/2}$,
\begin{eqnarray*}
  \Omega_i(z,\t) & = & (x_0^i(E_n,\t) + x_1^i(E_n,\t) z + \ldots) \, dz \\
  & = & \frac{1}{2} (x_0^i(E_n,\t) (\l - E_n)^{-1/2} + x_1^i(E_n,\t) + \ldots) \, d\l
\end{eqnarray*}
where we used that $dz = \frac{1}{2z} d\l$. Since by item (i) of Lemma \ref{chi1chi2}, $\Omega_1(E_n,\t) \neq 0$, it follows that $x_0^1(E_n,\t) \neq 0$. We now integrate $\Omega_i(\t)$ to get that
\begin{eqnarray*}
	\int_{E_n}^z \Omega_i(\t) & = & \int_{E_n}^\l \left( \frac{1}{2} \frac{x_0^i(E_n,\t)}{(\l - E_n)^{1/2}} + \frac{1}{2} x_1^i(E_n,\t) + \ldots \right) \, d\l \\
	& = & x_0^i(E_n,\t) (\l - E_n)^{1/2} + \frac{1}{2} x_1^i(E_n,\t) (\l - E_n) + \ldots
\end{eqnarray*}
is a value of the multi-valued function $J_i(p,\t)$. Then the $\t$-derivative of $J_i(p,\t)$ at $\t=0$ satisfies
	\[ \d J_i(p) = \left\{ -\frac{x_0^i(E_n,0)}{2} (\l - E_n)^{-1/2} + O(\l - E_n)^0 \right\} \delta E_n.
\]
Hence $\d J_i$ admits in $E_n$ a Laurent expansion and therefore is meromorphic near $E_n$. At $E_n$, it might have a pole of order $1$ with residue $r_i(n) \d E_n$ and $r_i(n) = -\frac{1}{2} x_0^i(E_n,0)$. Moreover $r_1(n) \neq 0$ as $x_0^1(E_n,0) \neq 0$ by the observation above.

To see that $\d J_i$ extends meromorphically to $\infty^+$ and $\infty^-$, use the expansions (\ref{intomega1}) and (\ref{intomega2}) to conclude that for $\l \to \infty$,
	\[ J_1(\l,\t) = \mp \left( \log \l + e_0(\t) + e_1(\t) \frac{1}{\l} + \ldots \right)
\]
and
	\[ J_2(\l,\t) = \mp (\l + f_0(\t) + \ldots).
\]
In view of the assumption that $\d e_0 = 0$ and $\d e_1 = 0$ it follows that $\d J_1(\l) = O\left( \frac{1}{\l^2} \right)$, and hence $\d J_1$ has a zero of order $2$ at $\infty^\pm$.
\end{proof}


The most important ingredient for the proof of Theorem \ref{krichevertheorem} is the following

\begin{prop} \label{dj1equiv0}
$\d J_1 \equiv 0$.
\end{prop}

To prove Proposition \ref{dj1equiv0} we first need to introduce an auxiliary function. For $p \in \mathcal{C}_E \setminus N_{\Omega_1}$, $d J_1(p) = \Omega_1(p) \neq 0$. Hence by the implicit function theorem, there exists a smooth curve $\t \mapsto q(\t) := q(\t,p)$ with $q(0) = p$ defined for $\t$ sufficiently close to zero so that $J_1(q(\t),\t) = J_1(p)$. In particular, for $p=E_n$ one has $q(\t) = E_n(\t)$. Then introduce for $p \in \mathcal{C}_E \setminus N_{\Omega_1}$
	\[ \d K(p) := \frac{d}{d\t} \Big|_{\t=0} J_2(q(\t),\t).
\]
As the periods of $\Omega_2$ are constant up to $O(\t^2)$ and $J_2(p,\t)$ is well defined up to half periods of $\Omega_2$, $\delta K$ is single-valued. Moreover, $\d K$ admits a meromorphic extension to $\Sigma_E$. Indeed, as $J_1(q(\t),\t) = J_1(p)$, one has for any $p \in \mathcal{C}_E \setminus N_{\Omega_1}$
	\[ \d J_1(p) + \langle \Omega_1(p), \d q \rangle = 0
\]
where $\langle \cdot, \cdot \rangle$ denotes the dual pairing between $T_p^* \Sigma_E$ and $T_p \Sigma_E$. Hence
	\[ \d K(p) = \frac{d}{d\t} \Big|_{\t=0} J_2(q(\t),\t) = \d J_2(p) + \langle \Omega_2(p), \d q \rangle
\]
leads to
\begin{equation} \label{dkpeqn}
	\d K(p) = \d J_2(p) - \frac{\Omega_2(p)}{\Omega_1(p)} \, \d J_1(p).
\end{equation}
By Lemma \ref{chi1chi2} we know that $\frac{\Omega_2(p)}{\Omega_1(p)}$ extends to a meromorphic function on $\Sigma_E$ with possible poles at the zeroes of $\Omega_1$. In view of Lemma \ref{propdtomega1}, $\d K$ admits a meromorphic extension to $\Sigma_E$.

\begin{lemma} \label{omega2equiv0}
$\d K \equiv 0$.
\end{lemma}

\begin{proof}[Proof of Lemma \ref{omega2equiv0}]
We show that, when counted with their orders, the number of poles of $\d K$ does not match the number of zeroes. First note that $\d K(E_n) = 0$ for any $1 \leq n \leq 2N$. Indeed, if $p=E_n$ for some $1 \leq n \leq 2N$, one has $q(\t) = E_n(\t)$ and hence for $i=1,2$, $J_i(E_n(\t), \t)$ contains zero for any $\t$, implying that $\d K(E_n) = 0$. On the other hand, by (\ref{dkpeqn}), the poles of $\d K$ in $\mathcal{C}_E$ are contained in the set $N_{\Omega_1}$ of zeroes of $\Omega_1$. By Lemma \ref{chi1chi2}, all these zeroes are simple and $|N_{\Omega_1}| = 2N-2$. Now let us investigate the values of $\d K$ at $\infty^+$ and $\infty^-$. Using the standard charts $z = \frac{1}{\l}$ we have by Lemma \ref{abeldiffexist}
	\[ \frac{\Omega_2(z)}{\Omega_1(z)} = O\left( \frac{1}{z} \right)
\]
and by Lemma \ref{propdtomega1}, $\d J_1(z) = O(z^2)$. Hence
	\[ \frac{\Omega_2(z)}{\Omega_1(z)} \d J_1(z) = O(z).
\]
It means that $\frac{\Omega_2(z)}{\Omega_1(z)} \d J_1(z)$ vanishes at $\infty^+$ and $\infty^-$. In addition, again by Lemma \ref{propdtomega1}, $\d J_2$ is holomorphic at $\infty^+$ and $\infty^-$. Alltogether we have shown that the meromorphic function $\d K$ has at least $2N$ zeroes and at most $2N-2$ poles (counted with their multiplicities). As $\Sigma_E$ is a compact surface it follows that $\d K \equiv 0$.
\end{proof}

\begin{proof}[Proof of Proposition \ref{dj1equiv0}]
By Lemma \ref{omega2equiv0}, formula (\ref{dkpeqn}) implies that
\begin{equation} \label{j1omega2j2omega1}
  \d J_1 \cdot \Omega_2 \equiv \d J_2 \cdot \Omega_1.
\end{equation}
By comparing poles and zeroes of $\d J_2 \cdot \Omega_1$ and $\d J_1 \cdot \Omega_2$ we want to conclude that $\d J_1 \equiv 0$ (and hence $\d J_2 \equiv 0$ as well). Indeed, by Lemma \ref{propdtomega1}, any pole of $\d J_1$ has to be a ramification point of $\Sigma_E$ and is of order $1$. By Lemma \ref{chi1chi2}, at least $2N-2$ zeroes of $\Omega_2$ are contained in $\mathcal{C}_E \setminus E$. We now have to distinguish between two cases. If $\Omega_2(E_n) \neq 0$ for any $1 \leq n \leq 2N$, then $\Omega_2$ has $2N$ zeroes and they are all contained in $\Sigma_E \setminus (E \cup \{ \infty^+, \infty^- \})$. By Lemma \ref{chi1chi2}, the zeroes of $\Omega_2$ cannot be zeroes of $\Omega_1$ and hence (\ref{j1omega2j2omega1}) implies that they must be zeroes of $\d J_2$. In addition, by Lemma \ref{propdtomega1}, $\d J_1$ vanishes at $\infty^\pm$ of order $2$ whereas $\Omega_2$ has a pole of order $2$. Hence $\d J_1 \cdot \Omega_2$ is holomorphic at $\infty^\pm$. By (\ref{j1omega2j2omega1}), $\d J_2 \cdot \Omega_1$ is then holomorphic at $\infty^\pm$. As $\Omega_1$ has a pole of order $1$ at $\infty^\pm$ it follows that $\d J_2$ vanishes at $\infty^\pm$. Alltogether, $\d J_2$ has at least $2N+2$ zeroes on $\Sigma_E$. On the other hand, by Lemma \ref{propdtomega1}, $\d J_2$ has at most $2N$ poles (all of them simple). As $\Sigma_E$ is a compact Riemann surface, the meromorphic function $\d J_2$ vanishes identically, and hence by (\ref{j1omega2j2omega1}), $\d J_1$ as well.

It remains to consider the case where there exists $1 \leq n \leq 2N$ so that $\Omega_2(E_n) = 0$. By Lemma \ref{propdtomega1}, $\d J_1$ is either holomorphic near $E_n$ or has a pole of order $1$. Hence $\d J_1 \cdot \Omega_2$ is holomorphic near $E_n$. By (\ref{j1omega2j2omega1}), $\d J_2 \cdot \Omega_1$ is then holomorphic at $E_n$ as well. By Lemma \ref{chi1chi2}, $\Omega_1(E_n) \neq 0$, hence $\d J_2$ is holomorphic near $E_n$. Again by Lemma \ref{propdtomega1}, we then see that $\d J_2$ has at most $2N-1$ poles in $\Sigma_E$. On the other hand, in view of Lemma \ref{chi1chi2}, $\d J_2$ has at least $2N-2$ zeroes in $\mathcal{C}_E \setminus E$. We have already seen that $\d J_2$ vanishes at $\infty^+$ and $\infty^-$. Hence $\d J_2$ has at least $2N$ zeroes and at most $2N-1$ poles in $\Sigma_E$. As $\Sigma_E$ is a compact Riemann surface, the meromorphic function $\d J_2$ vanishes identically, and so does $\d J_1$.
\end{proof}

\section{Formulas for the Toda frequencies} \label{freqtoda}

In this section we derive formulas for the frequencies of the periodic Toda lattice in terms of periods of the Abelian differential $\Omega_2$ introduced in section \ref{krichever}. These formulas will be used in an essential way to show that the frequency map is nondegenerate on $\R_{>0}^{N-1}$.

Introduce $\M^\bullet = \cup_{\a>0, \b \in \R} \Mba^\bullet$. As pointed out at the end of section \ref{prelim}, $\M^\bullet = \M \setminus \cup_{n=1}^{N-1} D_n$, i.e. for any $(b,a) \in \M^\bullet$, all the roots $(\l_i)_{1 \leq i \leq 2N}$ of $\Delta_\l^2(b,a) - 4$ are simple. As above, we list these roots in increasing order, $\l_1 < \l_2 < \ldots < \l_{2N}$. By Corollary \ref{mbacor}, $\M^\bullet$ is open and dense in $\M$. Given $(b,a) \in \M^\bullet$, denote by $\Sigma_{b,a}$ the Riemann surface $\Sigma_E$ with $E = (\l_1 < \ldots < \l_{2N})$, and scaling factor $\sigma = \a^{-2N}$, where $\a = \left( \prod_{i=1}^N a_i \right)^{1/N}$. In view of the product representation (\ref{delta2lrepr}) of $\Delta_\l^2(b,a)-4$, $\Sigma_{b,a} = \mathcal{C}_{b,a} \cup \{ \infty^+, \infty^- \}$, where 
\begin{equation} \label{algcurve}
\mathcal{C}_{b,a} := \{ (\l,z) \in \C^2 : z^2 = \Delta^2_{\l} (b, a) - 4 \}.
\end{equation}

To obtain a formula for the differential $\Omega_2$ we first consider an auxiliary differential.

\begin{lemma} \label{tildeomega2lemma}
Assume that $(b,a) \in \Mba^\bullet$ with $\b=0$. Then the differential
	\[ \tilde{\Omega}_2 = \frac{\l \dot{\Delta}_{\l}}{\sqrt{\Delta_{\l}^2-4}} \, d\l \]
	is holomorphic on $\mathcal{C}_{b,a}$ and has an expansion of the form $\left( \mp \frac{N}{z^2} + O(1) \right) dz$ at $\infty^\pm$ when expressed in the standard chart $z = \frac{1}{\l}$ of $\infty^\pm$.
\end{lemma}

The proof of Lemma \ref{tildeomega2lemma} is straightforward. For the convenience of the reader it is given in Appendix \ref{abeliandiff2}.

The Abelian differential $\tilde{\Omega}_2$ has to be appropriately normalized. For this purpose introduce the \emph{$\psi$-functions}. Let $(b,a) \in \M^\bullet$ and $1 \leq n \leq N-1$. Then there exists a unique polynomial $\psi_n(\l)$ 
of degree at most $N-2$ such that for any $1 \leq k \leq N-1$
\begin{equation} \label{psi}
\frac{1}{2\pi} \int_{c_k} \frac{\psi_n(\l)}{\sqrt{\Delta^2_\l-4}} \, d\l = \delta_{kn}.
\end{equation}
Here $(c_k)_{\1N1}$ denote the cycles on the canonical sheet $\Sigma_{b,a}^c$ of $\Sigma_{b,a}$ introduced at the beginning of section \ref{krichever}. For any $k \neq n$ it follows from (\ref{psi}) that
\begin{equation} \label{psiproperty}
\frac{1}{\pi} \int_{\l_{2k}}^{\l_{2k+1}} \frac{\psi_n(\l)}{\sqrt[+]{\Delta^2_\l - 4}} \, d\l = 0.
\end{equation}
As $(b,a) \in \M^\bullet$, $\g_k = \l_{2k+1} - \l_{2k} > 0$ for any $\1N1$ and hence in every gap $(\l_{2k}, \l_{2k+1})$ with $k \neq n$ the polynomial $\psi_n$ has a zero which we denote by $\sigma_k^n$. 
As $\psi_n(\l)$ is a polynomial of degree at most $N-2$, one has
\begin{equation} \label{psiprodrepr}
  \psi_n(\l) = M_n \prod_{1 \leq k \leq N-1 \atop k \neq n} (\l - \sigma_k^n),
\end{equation}
where $M_n \equiv M_n(b,a) \neq 0$. 
Clearly, the differential forms ($1 \leq n \leq N-1$)
\begin{equation} \label{zetandef}
  \zeta_n = \frac{\psi_n(\l)}{\sqrt{\Delta^2_\l - 4}} \, d\l
\end{equation}
are holomorphic on $\Sigma_{b,a} \setminus \{ \infty^+, \infty^- \}$. As the $\psi_n$'s are polynomials in $\l$ of degree at most $N-2$, they are also holomorphic at $\infty^+$ and $\infty^-$. Further, the action variables $I_n = I_n(b,a)$, $1 \leq n \leq N-1$, introduced in Theorem \ref{sumthm} for any $(b,a) \in \M$, are given by
\begin{equation} \label{actsba}
I_n = \frac{1}{2\pi} \int_{c_n} \l \frac{\dot{\Delta}_\l}{\sqrt{\Delta^2_\l-4}} \; d\l.
\end{equation}
They can be interpreted as period integrals of $\tilde{\Omega}_2$,
\begin{equation} \label{incnformel}
	I_n = \frac{1}{2\pi} \int_{c_n} \tilde{\Omega}_2 \qquad (1 \leq n \leq N-1).
\end{equation}

Now introduce the meromorphic differential
	\[ \Omega := \tilde{\Omega}_2 - \sum_{n=1}^{N-1} I_n \zeta_n
\]
with $(\zeta_n)_{1 \leq n \leq N-1}$ as given by (\ref{zetandef}).

\begin{lemma} \label{domega2formula}
Assume that $(b,a) \in \Mba^\bullet$ with $\b=0$ but $\a>0$ arbitrary. Then the meromorphic differentials $\Omega_2$ and $\Omega$ are related by $\Omega = -N \Omega_2$.
\end{lemma}

\begin{proof}
In view of the uniqueness statement of Lemma \ref{abeldiffexist} it suffices to show that $\Omega$ is a meromorphic differential so that (i) $\Omega$ is holomorphic on $\Sigma_{b,a} \setminus \{ \infty^+, \infty^- \}$; (ii) when expressed in the standard chart $\l = \frac{1}{z}$ near $\infty^+$ and $\infty^-$, $\Omega$ has a Laurent expansion of the form $\Omega = \left( \mp \frac{N}{z^2} + O(1) \right) dz$; (iii) $\int_{c_k} \Omega = 0$ for any $\1N1$.

Statements (i) and (ii) follow from Lemma \ref{tildeomega2lemma} and the above mentioned fact that the $\zeta_n$'s are holomorphic differentials on $\Sigma_{b,a}$. To see that the normalization conditions are satisfied, we use the identity (\ref{incnformel}) and the normalization conditions (\ref{psi}) to conclude that for any $\1N1$,
	\[ \int_{c_k} \Omega = \int_{c_k} \tilde{\Omega}_2 - \sum_{n=1}^{N-1} I_n \int_{c_k} \zeta_n = 2 \pi I_k - \sum_{n=1}^{N-1} I_n 2 \pi \delta_{kn} = 0,
\]
proving (iii).
\end{proof}

Recall that the Toda frequencies are given by
\begin{equation} \label{freqdef}
	\omega_n = \partial_{I_n} H_\a \qquad (1 \leq n \leq N-1),
\end{equation}
where $H_\a = H_\a(I_1, \ldots, I_{N-1})$ is, up to an additive constant given in Theorem \ref{sumthm}, the Hamiltonian of the Toda lattice expressed in terms of the action variables $I = (I_1, \ldots, I_{N-1})$ and the value $\a$ of the Casimir $C_2$. In particular, it follows that the frequencies $\omega_n$ ($1 \leq n \leq N-1$) are independent of $\b$. Without loss of generality we can therefore assume that $\b=0$. 
Expressing the element $(b,a) \in \M^\bullet$ with $\b=0$ in terms of the Birkhoff coordinates $(x,y)$ of Theorem \ref{sumthm} we may view $\Delta_\l$ as an analytic function of $\l$, $\a$, and $(x,y)$. As $\Delta$ is a spectral invariant, it is indeed an analytic function of $\l$, $\a$, and the action variables alone. Consider its gradient with respect to $I = (I_n)_{1 \leq n \leq N-1}$ and introduce the one-forms
\begin{equation} \label{etandef}
  \eta_n := -\frac{\partial_{I_n} \Delta}{\sqrt{\Delta^2 - 4}} \, d\l.
\end{equation}
These are holomorphic one-forms on $\Sigma_{b,a}$ except possibly at $\infty^\pm$. As
\begin{eqnarray*}
	\Delta_\l & = & 2 + \a^{-N} \prod_{j=1}^N (\l - \l_j^+) \\
	& = & 2 + \a^{-N} \l^N + \a^{-N} \! \left( \! \sum_{j=1}^N \l_j^+ \right) \l^{N-1} + O(\l^{N-2})
\end{eqnarray*}
and $\sum_{j=1}^N \l_j^+ = \sum_{n=1}^N b_n = -N \b = 0$ by assumption, $\partial_{I_n} \Delta$ is a polynomial in $\l$ of degree at most $N-2$, and hence $\eta_n$ is holomorphic at $\infty^+$ and $\infty^-$ as well.

In view of the definition of $\eta_n$,
	\[ \eta_n = \partial_{I_n} \left( \textrm{arcosh } \frac{\Delta_\l}{2} \right) \, d\l.
\]
To analyze $\eta_n$ near $\infty^+$ and $\infty^-$, we need to compute the asymptotic expansion of arcosh $\frac{\Delta_\l}{2}$ for $\l > \l_{2N}$ large. Denote by arcosh $x$ the positive branch of arcosh, i.e. arcosh $x > 0 \; \forall \, x > 1$. In Appendix \ref{asymtotic} we prove

\begin{prop} \label{arcoshlemma}
For any $(b,a) \in \Mba$, arcosh $\frac{\Delta_\l}{2}$ admits the asymptotic expansion ($\l \in \R, \l \to \infty$)
	\[ \textrm{arcosh } \frac{\Delta_\l}{2} = N \log \l - N \log \a + \frac{N \b}{\l} - \frac{H_{Toda}}{\l^2} + O(\l^{-3}).
\]
\end{prop}

Proposition \ref{arcoshlemma} leads to the following asymptotic expansion of $\eta_n$
\begin{equation} \label{etaomega}
  \eta_n = \pm \left( \frac{\omega_n}{\l^2} + O(\l^{-3}) \right) d\l
\end{equation}
with respect to the local coordinate $\l$ near $\infty^\pm$.

Finally, we show that for any $1 \leq n \leq N-1$, the holomorphic one-form $\eta_n$ coincides with the one-form $\zeta_n$ introduced earlier.

\begin{lemma} \label{etazetalemma}
For any $(b,a) \in \M^\bullet$ with $\b=0$ and any $1 \leq n \leq N-1$,
\begin{equation} \label{etazeta}
  \eta_n = \zeta_n.
\end{equation}
\end{lemma}

\begin{proof}
Let $1 \leq n \leq N-1$ be fixed. We already know that $\eta_n$ is a holomorphic one-form on $\Sigma_{b,a}$. To show that it coincides with $\zeta_n$ it suffices to prove that it satisfies the normalizing conditions (\ref{psi}),
	\[ \frac{1}{2\pi} \int_{c_k} \eta_n = \delta_{nk} \qquad \forall \, \1N1
\]
or
\begin{equation} \label{etazetaid}
  \int_{\pi(c_k)} \frac{\partial_{I_n} \Delta_\l}{\sqrt[c]{\Delta_\l^2 - 4}} \, d\l = -2 \pi \delta_{nk} \qquad \forall \, \1N1,
\end{equation}
where $\pi: \Sigma_{b,a} \to \C \cup \{ \infty \}$ is the projection introduced at the beginning of section \ref{krichever}. Note that the principal branch of the logarithm
\begin{displaymath}
  \kappa(\l) = \log \left( (-1)^{N-k} \left( \Delta_\l - \sqrt[c]{\Delta_\l^2 - 4} \right) \right)
\end{displaymath}
is well-defined for $\l$ near $\pi(c_k)$ and depends analytically on $(I_n)_{1 \leq n \leq N-1}$. By a straightworward computation, for $\l$ near $\pi(c_k)$
	\[ \partial_{I_n} \kappa = \frac{\partial_{I_n} \Delta - \frac{\Delta \, \partial_{I_n} \Delta}{\sqrt[c]{\Delta^2 - 4}}}{\Delta - \sqrt[c]{\Delta^2 - 4}} = - \frac{\partial_{I_n} \Delta}{\sqrt[c]{\Delta^2 - 4}}.
\]
Hence the left hand side of the identity (\ref{etazetaid}) can be computed to be
\begin{equation} \label{etazetaid2}
	\int_{\pi(c_k)} \frac{\partial_{I_n} \Delta}{\sqrt[c]{\Delta^2 - 4}} \, d\l = - \int_{\pi(c_k)} \partial_{I_n} \kappa(\l) d\l = -\partial_{I_n} \int_{\pi(c_k)} \kappa(\l) d\l.
\end{equation}
On the other hand, for $\l$ near $\pi(c_k)$,
	\[ \partial_\l \kappa = \frac{\dot{\Delta} - \frac{\Delta \, \dot{\Delta}}{\sqrt[c]{\Delta^2 - 4}}}{\Delta - \sqrt[c]{\Delta^2 - 4}} = - \frac{\dot{\Delta}}{\sqrt[c]{\Delta^2 - 4}}
\]
and thus, by integration by parts,
	\[ 2 \pi I_k = \int_{\pi(c_k)} \l \frac{\dot{\Delta}}{\sqrt[c]{\Delta^2 - 4}} = \int_{\pi(c_k)} \l (-\partial_\l \kappa) = \int_{\pi(c_k)} \kappa(\l) d\l.
\]
Combined with (\ref{etazetaid2}), we get the claimed identity (\ref{etazetaid}).
\end{proof}

\begin{theorem} \label{freqintrepr}
For any $(b,a) \in \M^\bullet$ and any $1 \leq n \leq N-1$, the Toda frequency $\omega_n = \partial_{I_n} H_\a$ satisfies
\begin{equation} \label{omegakformel}
  \omega_n = \frac{i}{2} \int_{d_n} \Omega_2.
\end{equation}
\end{theorem}


\begin{proof}
To prove (\ref{omegakformel}) we use the Riemann bilinear relations. Fix $1 \leq n \leq N-1$ and $(b,a) \in \M$. We have already observed that $\omega_n$ does not depend on $\b$. Without loss of generality we therefore can assume that $\b=0$ for the given element $(b,a) \in \M^\bullet$. Combining (\ref{etaomega}) with Lemma \ref{etazetalemma} we conclude that for $\l$ near $\infty^\pm$
	\[ \zeta_n = \pm \left( \frac{\omega_n}{\l^2} + O(\l^{-3}) \right) d\l.
\]
When expressed in the standard chart $\l = \frac{1}{z}$, we have $\zeta_n = f_n^\pm(z) dz$ for $z$ near $0$ with
\begin{equation} \label{fnzomega}
  f_n^\pm(z) = \mp \omega_n + O(z).
\end{equation}

By the Riemann bilinear relations, applied to $\Omega$ and $\zeta_n$, we then get (cf. e.g. \cite{grha}, p. 241)
	\[ \sum_{k=1}^{N-1} \left( \int_{c_k} \zeta_n \int_{d_k} \Omega - \int_{d_k} \zeta_n \int_{c_k} \Omega \right) = 2 \pi i \left( -N f_n^+(0) + N f_n^-(0) \right) = 4 \pi i N \omega_n.
\]
Using that $\int_{c_k} \Omega = 0$ and $\int_{c_k} \zeta_n = 2 \pi \delta_{nk}$ for any $\1N1$ the left hand side of the above identity equals $2 \pi \int_{d_n} \Omega$. Hence
\begin{equation} \label{intdnomega2}
	\int_{d_n} \Omega = 2 N i \omega_n.
\end{equation}
By Lemma \ref{domega2formula} and the assumption $\b=0$, the one-forms $\Omega$ and $\Omega_2$ are related by $\Omega = -N \Omega_2$. Together with (\ref{intdnomega2}) 
the claimed identity then follows.
\end{proof} 
We remark that in the seventies, Its and Matveev have obtained a formula for the frequencies of the KdV equation similar to (\ref{omegakformel}) - see e.g. \cite{kapo} for a detailed exposition. For the Toda lattice, computations similar to the ones in the proof of Theorem \ref{freqintrepr} can be found in \cite{teschl}.

Finally we note that $\int_{d_n} \Omega$ can be written as
	\[ \int_{d_n} \Omega = 2 \int_{\l_1}^{\l_{2n}} \frac{\l \dot{\Delta}_{\l}}{\sqrt[c]{\Delta_{\l - i 0}^2 - 4}} \, d\l - 2 \sum_{k=1}^{N-1} I_k \int_{\l_1}^{\l_{2n}} \frac{\psi_k(\l)}{\sqrt[c]{\Delta_{\l - i 0}^2 - 4}} \, d\l.
\]
As $\M^\bullet$ is dense in $\M$ it then follows that for any $(b,a) \in \M$
	\[ \omega(I;\a) = \frac{1}{N i} \left( \int_{\l_1}^{\l_{2n}} \! \frac{\l \dot{\Delta}_{\l}}{\sqrt[c]{\Delta_{\l \!-\! i 0}^2 \!-\! 4}} \, d\l - \sum_{k=1}^{N-1} I_k \int_{\l_1}^{\l_{2n}} \! \frac{\psi_k(\l)}{\sqrt[c]{\Delta_{\l \!-\! i 0}^2 \!-\! 4}} \, d\l \right) \Bigg|_{(b + \b 1_N, a)},
\]
where $I = I(b,a)$ is given by (\ref{actsba}), $\a = \left( \prod_{i=1}^N a_i \right)^{1/N}$, $\b = -\frac{1}{N} \sum_{k=1}^N b_k$, and $1_N \in \R^N$ is the vector $1_N = (1, \ldots, 1)$.

\section{Proof of Theorem \ref{nekhtodatheorem}} \label{mainsection}

In this section we prove Theorem \ref{nekhtodatheorem}. The main ingredients are the Birkhoff normal form of the Toda lattice (Theorem \ref{bnf4prop} and Corollary \ref{strictconvexlocal}) and Krichever's theorem (Theorem \ref{krichevertheorem}).

We begin by computing the components of the period map $\mathcal{F}$, defined by (\ref{mathcalfdef}) in section \ref{krichever}, for sequences $\l_1 < \ldots < \l_{2N}$ given by the spectrum of the $2N \times 2N$-Jacobi matrix $L^+((b,b),(a,a))$ with $(b,a) \in \M^\bullet$. To compute the period $\int_{d_n} \Omega_1$ of $\Omega_1$ we need

\begin{lemma} \label{omega1lemma}
For any $(b,a) \in \M^\bullet$
\begin{equation} \label{omega1formel}
	\Omega_1 = -\frac{1}{N} \frac{\dot{\Delta}}{\sqrt{\Delta^2 - 4}} d\l.
\end{equation}
\end{lemma}

\begin{proof}
Let $(b,a) \in \M^\bullet$ be given. Clearly, $-\frac{1}{N} \frac{\dot{\Delta}}{\sqrt{\Delta^2 - 4}} d\l$ is a holomorphic one-form on $\Sigma_{b,a} \setminus \{ \infty^+, \infty^- \}$. We claim that it has poles of order $1$ at $\infty^\pm$. Indeed, in the standard chart $z = \frac{1}{\l}$ near $\infty^\pm$ one has
\begin{eqnarray*}
	\frac{\dot{\Delta}}{\sqrt{\Delta^2 - 4}} d\l & = & N \, \frac{z^{-(N-1)} \prod_{i=1}^{N-1} (1 - \dot{\l}_i z)}{z^{-N} \sqrt{\prod_{i=1}^{2N} (1 - \l_i z)}} \frac{dz}{-z^2} \\
	& = & \mp \left( \frac{N}{z} + O(1) \right) dz.
\end{eqnarray*}
In view of the uniqueness statement of Lemma \ref{abeldiffexist}, it remains to show that the normalization conditions (\ref{omegainorm}) are satisfied. One computes for any $\1N1$
\begin{eqnarray*}
	\int_{c_k} \frac{\dot{\Delta}}{\sqrt{\Delta^2 - 4}} \, d\l & = & \int_{\pi(c_k)} \frac{\dot{\Delta}}{\sqrt[c]{\Delta^2 - 4}} \, d\l \\
	& = & 2 \, \textrm{arcosh } \left( (-1)^{N-k} \frac{\Delta_\l}{2} \right) \Bigg|_{\l = \l_{2k}}^{\l = \l_{2k+1}} \\ & = & 0.
\end{eqnarray*}
Now identity (\ref{omega1formel}) follows from the uniqueness statement of Lemma \ref{abeldiffexist}.
\end{proof}

\begin{cor} \label{intdnomega1}
For any $(b,a) \in \M^\bullet$
\begin{equation}
	\int_{d_n} \Omega_1 = \frac{2 \pi n}{N} i.
\end{equation}
\end{cor}

\begin{proof}
By Lemma \ref{omega1lemma} and the normalization conditions (\ref{omegainorm}) one gets for any $1 \leq n \leq N-1$
	\[ \int_{d_n} \Omega_1 = -\frac{2}{N} \left( \int_{\l_1}^{\l_2} + \int_{\l_3}^{\l_4} + \ldots + \int_{\l_{2n-1}}^{\l_{2n}} \right) \frac{\dot{\Delta}_\l}{\sqrt[c]{\Delta^2_{\l - i0} - 4}} \, d\l.
\]
For any $\l_{2k-1} \leq \l \leq \l_{2k}$,
	\[ \frac{\dot{\Delta}_\l}{\sqrt[c]{\Delta^2_{\l - i0} - 4}} = \frac{(-1)^{N-k} \dot{\Delta}/2}{i \sqrt[+]{1 - (\Delta/2)^2}} \, d\l = \frac{1}{i} \, \partial_\l \left( \arcsin \left((-1)^{N-k} \frac{\Delta_\l}{2} \right) \right)
\]
and thus
	\[ \int_{d_n} \Omega_1 = -\frac{2}{Ni} \sum_{k=1}^n \arcsin \left( (-1)^{N-k} \frac{\Delta_\l}{2} \right) \Bigg|_{\l_{2k-1}}^{\l_{2k}} = \frac{2 n \pi}{N} \, i,
\]
as claimed.
\end{proof}

To obtain the last two components of the period map $\mathcal{F}$ we compute the asymptotics of $\int_{\l_{2N}}^\l \Omega_1$ with respect to the local coordinate $\l$ near $\infty^\pm$. By Lemma \ref{omega1lemma}, we get, near $\infty^\pm$,
	\[ \int_{\l_{2N}}^\l \Omega_1 = -\frac{1}{N} \int_{\l_{2N}}^\l \frac{\dot{\Delta}_\l}{\sqrt{\Delta_\l^2 - 4}} \, d\l = \mp \frac{1}{N} \, \textrm{arcosh} \frac{\Delta_\l}{2}
\]
and hence by Proposition \ref{arcoshlemma},
\begin{equation} \label{intl2nlomega1}
	\int_{\l_{2N}}^\l \Omega_1 = \mp \left( \log \l - \log \a + \frac{\b}{\l} + O\left( \l^{-2} \right) \right),
\end{equation}
or, in the notation of section \ref{krichever}, $e_1 = \b$ and $e_0 = -\log \a$. Taking into account that by Theorem \ref{freqintrepr},
	\[ \int_{d_n} \Omega_2 = -2 i \omega_n \qquad \forall \, 1 \leq n \leq N-1
\]
and that $\omega_n = \partial_{I_n} H_\a$ we therefore have proved

\begin{prop} \label{mathcalflambdaprop}
For any $(b,a) \in \M^\bullet$
	\[ \mathcal{F} (\l_1 < \ldots < \l_{2N}) = \left( \left( \frac{2 n \pi i}{N}, -2i \, \partial_{I_n} H_\a \right)_{1 \leq n \leq N-1}, \b, -\log \a \right).
\]
\end{prop}

Next we define the map
\begin{eqnarray*}
	\Lambda: \quad \R_{>0}^{N-1} \times \R \times \R_{>0} & \to & \R^{2N} \\
	((I_n)_{1 \leq n \leq N-1}, \b, \a) & \mapsto & (\l_n)_{1 \leq n \leq 2N}
\end{eqnarray*}
where $(\l_n)_{1 \leq n \leq 2N}$ is the spectrum of the Jacobi matrix $L^+((b,b),(a,a))$ and $(b,a) \in \M$ is determined by the Birkhoff map $\Phi$ (cf. Theorem \ref{sumthm}),
	\[ (b,a) = \Phi^{-1} ((\sqrt[+]{2 I_n},0)_{1 \leq n \leq N-1}, \b, \a).
\]
Clearly, $\Lambda$ is $1-1$ and as $I_n > 0$ for any $1 \leq n \leq N-1$, $\Lambda$ is smooth. On its image, the inverse $\Lambda^{-1}$ of $\Lambda$ can be explicitly computed. In view of (\ref{actsba}) and (\ref{intl2nlomega1}) one has for any $(\l_n)_{1 \leq n \leq 2N} \in \textrm{im} \, \Lambda$,
	\[ \Lambda^{-1} ((\l_n)_{1 \leq n \leq 2N}) = \left( \left( \frac{1}{2\pi} \int_{c_n} \l \, \frac{\dot{f}_\l}{\sqrt{f_\l^2 - 4}} \, d\l \right)_{1 \leq n \leq N-1}, e_1, \exp(-e_0) \right),
\]
where in view of (\ref{spectgeordnet}), $f_\l = f(\l)$ is given by $f(\l) = 2 + \prod_{i=1}^N (\l - \l_i^+)$ with $\l_N^+ = \l_{2N}$, $\l_{N-1}^+ = \l_{2N-3}$, $\l_{N-2}^+ = \l_{2N-4}$, $\l_{N-3}^+ = \l_{2N-7}$, \ldots, and $e_1$, $e_0$ are the coefficients in the expansion (\ref{intomega1}) of the differential form $\Omega_1$ on the Riemann surface $\Sigma_E$ with $E = (\l_1 < \ldots < \l_{2N})$ and scaling parameter $\sigma=1$. (Note that by Lemma \ref{abeldiffexist}, $\Omega_1$ is independent of the scaling factor $\sigma$.) Hence we have shown

\begin{prop} \label{Lambdaprop}
The map
	\[ \Lambda: \R_{>0}^{N-1} \times \R \times \R_{>0} \to \R^{2N}
\]
is a smooth embedding.
\end{prop}

With these preparations we are now ready to prove Theorem \ref{nekhtodatheorem}.

\begin{proof}[Proof of Theorem \ref{nekhtodatheorem}]
In view of Proposition \ref{mathcalflambdaprop}, the composition $\mathcal{F} \circ \Lambda: \R_{>0}^{N\!-\!1} \times \R \times \R_{>0} \to \R^{2N}$ is given by
\begin{equation} \label{FcircL}
	\mathcal{F} \circ \Lambda \left( (I_n)_{1 \leq n \leq N-1}, \b, \a \right) = \left( \left( \frac{2 \pi n i}{N}, -2i  \, \partial_{I_n} H_\a \right)_{1 \leq n \leq N-1}, \, \b, \, -\log \a \right).
\end{equation}
By Theorem \ref{krichevertheorem}, $\mathcal{F}$ is a local diffeomorphism, and by Proposition \ref{Lambdaprop}, $\Lambda$ is a smooth embedding. Hence $\mathcal{F} \circ \Lambda$ is an embedding. Therefore, at each point $\left( (I_n)_{1 \leq n \leq N-1}, \b, \a \right)$ the differential $d(\mathcal{F} \circ \Lambda)$ has rank $N+1$. By (\ref{FcircL}) it is a $2N \times (N+1)$-matrix of the form
	\[ \left( \begin{array}{c|c|c}
&&\\
0_{(N-1) \times (N-1)} & 0_{(N-1) \times 1} & 0_{(N-1) \times 1} \\
&&\\
\hline
&&\\
(-2i \partial_{I_n} \partial_{I_l} H_\a)_{1 \leq n,l \leq N-1} & 0_{(N-1) \times 1} & 0_{(N-1) \times 1} \\
&&\\
\hline
0_{1 \times (N-1)} & 1 & 0 \\
\hline
\ldots & 0 & -\a^{-1}
\end{array} \right)
\]
where $0_{N_1 \times N_2}$ denotes the $N_1 \times N_2$-matrix with all entries $0$. Hence the rank of the $(N-1) \times (N-1)$-matrix $\left( \frac{\partial^2 H_\a}{\partial I_n \partial I_l} \right)_{1 \leq n,l \leq N-1}$ has to be $N-1$. This proves Theorem \ref{nekhtodatheorem}.
\end{proof}

\appendix

\section{Appendix: Proof of Lemma \ref{tildeomega2lemma}} \label{abeliandiff2}

By (\ref{dotlrepr}), $\dot{\Delta}_\l$ admits the product representation
	\[ \dot{\Delta}_\l = N \a^{-N} \prod_{n=1}^{N-1} (\l - \dot{\l}_n),
\]
where the roots $(\dot{\l}_n)_{1 \leq n \leq N-1}$ of $\dot{\Delta}$, when listed in increasing order, satisfy $\l_{2n} < \dot{\l}_n < \l_{2n+1}$ for any $1 \leq n \leq N-1$. Hence
\begin{equation} \label{lambdadeltalambda}
  \tilde{\Omega}_2 = N \frac{\l \, (\l - \dot{\l}_1) \cdots (\l - \dot{\l}_{N-1})}{\sqrt{(\l - \l_1) \cdots (\l - \l_{2N})}} \, d\l.
\end{equation}
It is clear that $\tilde{\Omega}_2$ is holomorphic on the set $\Sigma_{b,a} \setminus \{ \infty^+, \infty^- \}$. 
In the standard chart $\l = \frac{1}{z}$ at $\infty^+$ one has
	\[ \tilde{\Omega}_2 = N \frac{(1 - \dot{\l}_1 z) \cdots (1 - \dot{\l}_{N-1} z)}{\sqrt[+]{(1 - \l_1 z) \cdots (1 - \l_{2N} z)}} \cdot \frac{dz}{-z^2}.
\]
Using that $(1 - \l_n z)^{-\frac{1}{2}} = 1 + \frac{1}{2} \l_n z + O(z^2)$ near $z=0$ one gets
\begin{eqnarray*}
	\tilde{\Omega}_2 & = & -\frac{N}{z^2} \left( 1 - \left( \sum_{n=1}^{N-1} \dot{\l}_n \right) z + O(z^2) \right) \cdot \left( 1 + \frac{1}{2} \left( \sum_{n=1}^{2N} \l_n \right) z + O(z^2) \right) dz \\
	& = & \left( -N \frac{1}{z^2} + N \left( \sum_{n=1}^{N-1} \dot{\l}_n - \frac{1}{2} \sum_{n=1}^{2N} \l_n \right) \frac{1}{z} + O(1) \right) dz.
\end{eqnarray*}
By (\ref{lambdadotlambda}), one has
	\[ \sum_{n=1}^{N-1} \dot{\l}_n = \frac{N-1}{2N} \sum_{n=1}^{2N} \l_n = \frac{N-1}{N} \sum_{n=1}^N b_n
 = -(N-1) \b \]
Hence the coefficient of $\frac{1}{z}$ in the expansion above equals
	\[ N \left( \sum_{n=1}^{N-1} \dot{\l}_n - \frac{1}{2} \sum_{n=1}^{2N} \l_n \right) = N (-(N-1) \b + N \b) = N \b
\]
which by assumption equals zero. Alltogether we have proved that with respect to the standard chart $\l = \frac{1}{z}$ at $\infty^+$,
	\[ \tilde{\Omega}_2 = \left( -\frac{N}{z^2} + O(1) \right) dz.
\]
By a similar computation one sees that in the standard chart $\l = \frac{1}{z}$ at $\infty^-$, one has $\tilde{\Omega}_2 = \left( \frac{N}{z^2} + O(1) \right) dz$. This completes the proof of Lemma \ref{tildeomega2lemma}.

\section{Appendix: Proof of Proposition \ref{arcoshlemma}} \label{asymtotic}

To prove Proposition \ref{arcoshlemma} we first need to derive some auxiliary results. Let $(b,a) \in \Mba$ and assume that $\l > \l_{2N}$ in the sequel. Recall that the Floquet multipliers associated to the difference equation (\ref{diff}) are defined as the eigenvalues of the monodromy matrix
\begin{displaymath}
  \left( \begin{array}{cc} y_1(N,\l) & y_2(N,\l) \\ y_1(N+1,\l) & y_2(N+1,\l) \end{array} \right).
\end{displaymath}
Using the Wronskian identity, one sees that the characteristic polynomial of the monodromy matrix is given by $1 - \Delta_\l \xi + \xi^2$, hence the Floquet multipliers are $\xi_\pm(\l) = \frac{\Delta_\l}{2} \pm \frac{1}{2} \sqrt{\Delta_\l^2 - 4}$. As $\Delta_\l > 2$ for $\l > \l_{2N}$, $\xi_\pm(\l)$ are real valued and satisfy $\xi_+(\l) > 1 > \xi_-(\l) > 0$ as well as $\xi_+(\l) \cdot \xi_-(\l) = 1$. Solutions of (\ref{diff}) corresponding to the Floquet multiplier
\begin{displaymath}
  w(\l) \equiv \xi_+(\l) = \frac{\Delta_\l}{2} + \frac{1}{2} \sqrt{\Delta_\l^2 - 4}
\end{displaymath}
are thus expanding. On the other hand, as $\log (x + \sqrt{x^2-1}) = \textrm{arcosh } x$ for $x>1$ one has
	\[ \log \left( \frac{\Delta_\l}{2} + \frac{1}{2} \sqrt{\Delta_\l^2 - 4} \right) = \textrm{arcosh} \, \frac{\Delta_\l}{2}
\]
and therefore
\begin{displaymath}
  \log w(\l) = \textrm{arcosh} \, \frac{\Delta_\l}{2}.
\end{displaymath}
For any $\l > \l_{2N}$ denote by $(u(n, \l))_{n \in \Z}$ a solution of (\ref{diff}) satisfying
	\[ u(n+N,\l) = w(\l) u(n,\l) \quad \forall \; n \in \Z.
\]
We claim that
\begin{equation} \label{unlneq0}
	u(n,\l) \neq 0 \quad \forall \; n \in \Z.
\end{equation}
Indeed, if there were $k \in \Z$ with $u(k,\l) = 0$, then $\l$ would be an eigenvalue of $L_2(S^k(b,a))$ where $S^k(b,a)$ denotes the shifted element
	\[ S^k(b,a) := (b_{n+k},a_{n+k})_{1 \leq n \leq N} \in \M
\]
and $L_2(b,a)$ denotes the $(N-1) \times (N-1)$ Jacobi matrix given by
	\[ \left( \begin{array}{ccccc}
b_2 & a_2 & 0 & \ldots & 0  \\
a_2 & \ddots & \ddots & \ddots & \vdots  \\
0 & \ddots & \ddots & \ddots & 0 \\
\vdots & \ddots & \ddots & \ddots & a_{N-1} \\
0 & \ldots & 0 & a_{N-1} & b_N \\
\end{array} \right).
\]
However, 
	\[ \textrm{spec} \, L^\pm(S^{k}(b,a)) = \textrm{spec} \, L^\pm(b,a),
\]
and spec$(L_2(S^{k}(b,a)))$ is bounded by max $L^+(S^{k}(b,a))$ (cf. \cite{moer} or \cite{ahtk1}). This leads to a contradiction, and (\ref{unlneq0}) is proved. Hence the solution $(u(n,\l))_{n \in \Z}$ can always be normalized by $u(0,\l) = 1$. Then $w(\l) = u(N,\l) / u(0,\l) = u(N,\l)$, or
\begin{equation} \label{arcoslog}
  \textrm{arcosh} \, \frac{\Delta_\l}{2} = \log u(N,\l).
\end{equation}

\begin{proof}[Proof of Proposition \ref{arcoshlemma}]
Let $\l > \l_{2N}$ and write $u(n) = u(n,\l)$. In view of (\ref{unlneq0}) we may define
\begin{equation} \label{phidef}
  \phi(n) \equiv \phi(n,\l) := \frac{u(n+1)}{u(n)}, \quad n \in \Z.
\end{equation}
One verifies that $\phi(n)$ satisfies the discrete Riccati equation (cf. e.g. \cite{teschl})
\begin{equation} \label{riccati}
  a_n \phi(n) \phi(n-1) + (b_n - \l) \phi(n-1) + a_{n-1} = 0.
\end{equation}
In the case $b_n = -\b$ and $a_n = \a$ for any $n \in \Z$, $u(n,\l)_{n \in \Z}$ can be computed explicitly. Indeed, making the ansatz $u(n) = e^{\kappa n}$, one concludes that $\phi(n) \equiv e^\kappa$ is given by
\begin{displaymath}
 e^\kappa = \frac{\l + \b}{2\a} \left( 1 + \sqrt{1 - \left( \frac{2\a}{\l+\b} \right)^2} \right).
\end{displaymath}
Hence for $\phi(n)$ one gets the expansion
\begin{equation} \label{phinl}
	\phi(n, \l) = \frac{\l}{\a} + \frac{\b}{\a} + O(\l^{-1}), \; \textrm{as} \, \l \to \infty.
\end{equation}
In the case of an arbitrary element $(b,a) \in \Mba$, (\ref{phinl}) suggests to make the ansatz
\begin{equation} \label{phiansatz}
  \phi(n,\l) = \frac{\l}{a_n} - \frac{b_n}{a_n} + \frac{1}{a_n} \sum_{k=1}^\infty \frac{\phi_k(n)}{\l^k}.
\end{equation}
Substituting this ansatz into (\ref{riccati}) one gets by comparison of coefficients

\begin{equation} \label{phi1phi2formel}
  \phi_1(n) = -a_{n-1}^2 \quad \forall n \in \Z.
\end{equation}

By (\ref{phidef}), we have $u(N,\l) = \prod_{n=0}^{N-1} \phi(n,\l)$, and thus by (\ref{arcoslog}),
\begin{displaymath}
  \textrm{arcosh } \frac{\Delta_\l}{2} = \sum_{n=0}^{N-1} \log \phi(n,\l).
\end{displaymath}
In view of the asymptotic expansion (\ref{phiansatz}) and the values of the coefficients $\phi_1(n)$ 
given by (\ref{phi1phi2formel}) it then follows that
\begin{displaymath}
  \textrm{arcosh } \frac{\Delta_\l}{2} = \sum_{n=0}^{N-1} \log \frac{\l}{a_n} + \sum_{n=0}^{N-1} \log \left( 1 - \frac{b_n}{\l} - \frac{a_{n-1}^2}{\l^2} 
  + O(\l^{-3}) \right).
\end{displaymath}
Note that
\begin{equation} \label{casimirlambda}
  \sum_{n=0}^{N-1} \log \frac{\l}{a_n} = N \log \l - \log \prod_{n=0}^{N-1} a_n = N \log \l - N \log \a.
\end{equation}
Using $\log(1+x) = x - \frac{1}{2} x^2 + O(x^3)$, one sees that
	\[ \sum_{n=0}^{N-1} \log \left( 1 - \frac{b_n}{\l} - \frac{a_{n-1}^2}{\l^2} 
  + O(\l^{-3}) \right) = \frac{N \b}{\l} - \frac{1}{\l^2} \sum_{n=0}^{N-1} \left( a_{n-1}^2 + \frac{1}{2} b_n^2 \right) + O(\l^{-3}),
\]
which by (\ref{htodaflaschka}) equals
\begin{equation} \label{blabla}
	\frac{N \b}{\l} - \frac{1}{\l^2} H_{Toda}(b,a) + O(\l^{-3}).
\end{equation}
Combining (\ref{casimirlambda}) and (\ref{blabla}) we get the claimed expansion
\begin{displaymath}
  \textrm{arcosh } \frac{\Delta_\l}{2} = N \log \l - N \log \a + \frac{N \b}{\l} - \frac{1}{\l^2} H_{Toda} + O(\l^{-3}).
\end{displaymath}
This completes the proof of Proposition \ref{arcoshlemma}.
\end{proof}

\end{document}